\documentclass[10pt]{IEEEtran}

\usepackage{graphicx}
\usepackage{tikz}
\usepackage{pgfplots}
\pgfplotsset{compat=1.5}

\usetikzlibrary{calc}
\usetikzlibrary{spy}

\usepackage{amsfonts}
\usepackage{amsmath}
\usepackage{amsthm}
\usepackage{mathtools}
\usepackage{calrsfs}
\usepackage{manfnt}

\usepackage[mathscr]{euscript}

\usepackage[noadjust]{cite}

\theoremstyle{definition}

\newtheorem{theorem}{Theorem}

\newtheorem{proposition}{Proposition}
\newtheorem{example}{Example}

  \title{On Lossless Causal Compression \\ of Periodic Signals}
\author{
Jan Maximilian Montenbruck, 
Shen Zeng
\thanks{
Shen Zeng is with the School of Engineering \& Applied Science, Washington University in St.\ Louis, St.\ Louis, MO, USA.} 
}

\begin{document}
\maketitle
\thispagestyle{empty}
\pagestyle{empty}
\begin{abstract}
We present and study a scheme for lossless causal compression of periodic real-valued signals. In particular, our technique compresses a vector-valued signal to a scalar-valued signal by mixing it with another periodic signal. The conditions for being able to reconstruct the original signal then amount to certain non-resonances between the periods of the two signals. The proposed compression scheme turns out to implicitly be inherent to communication networks with round-robin scheduling and digital photography with active pixel sensors.
\end{abstract}

\section{Introduction}

With the increasing desire for ubiquitous availability of data, particularly digital media, the need for transmission of large amounts of data is nowadays increasing as well. This causes bandwidth requirements to continuously grow in all domains and thus necessitates signal compression techniques. These techniques have the purpose of reducing the amount of bits required to represent a signal. 

In signal compression, one distinguishes between techniques which are reversible, i.e., for which the original signal can be uniquely reconstructed, and, on the other hand, irreversible schemes. One refers to the former, reversible compression schemes as being \emph{lossless}. Lossless compression schemes are preferred whenever the quality of the reconstructed signal is important or when the reconstructed signal should be manipulated for further usage. For words drawn from finite alphabets, the so-called \emph{entropy coding} techniques, such as Shannon-Fano coding, Huffman coding, arithmetic coding, or dictionary coding, are well-established (cf. \cite[section 9]{Gersho1992}). For real-valued signals, particularly periodic signals, it is common to compute the Fourier coefficients of the signal in order to then transmit those coefficients, from which the signal can be reconstructed. This scheme is, however, not causal, i.e., computation of the Fourier coefficients requires availability of the entire signal. This implies that the technique cannot be applied \emph{online}, i.e., during signal transmission. 

Herein, we will present and discuss a scheme for lossless causal compression of periodic signals. Signal compression will thereby be realized by mixing the signal with another periodic signal. We will reveal that this concept is also inherent to technical applications such as communication networks with round-robin scheduling or digital photography with active pixel sensors. Since we will explicitly discuss reconstruction techniques for the compressed signals, our results will thus also allow for signal reconstruction in those applications.

\newpage

Throughout the paper, we consider periodic vector-valued signals $x$ which must be compressed to scalar-valued signals $y$ in a lossless fashion, i.e. such that it will thereafter be possible to uniquely reconstruct $x$ from $y$. This compression will be realized instantaneously, viz.\ by mixing the channels of $x$ with the channels of another vector-valued periodic signal $c$. By \emph{mixing} the channels of $x$ and $c$, we mean multiplying each scalar-valued $i$th signal $x_{i}$ in $x$ with the corresponding scalar-valued $i$th signal $c_{i}$ in $c$. The resulting scalar-valued signals are summed up, thereby yielding the compressed signal $y$. It shall be emphasized that this compression is \emph{causal}, i.e., values of $y$ up to a certain time cannot depend on values of $x$ or $c$ which only occur after that time.

More formally, let the signals $x$ and $c$ both evolve in $\mathbb{R}^{n}$. Then, at time $t$, the compressed signal $y$ results from taking the inner product $\langle c\left(t\right),x\left(t\right)\rangle$ of $c\left(t\right)$ with $x\left(t\right)$, i.e.
\begin{equation}
 y\left(t\right)=\langle c\left(t\right),x\left(t\right)\rangle. \label{eq:compression}
\end{equation}
Our interest in this compression stems not only from its reversibility and causality, but, in particular, also from its implicit occurrence in technical applications. In the following, we outline how \eqref{eq:compression} is, for instance, inherent to communication networks with round-robin scheduling and digital photography with active pixel sensors.

In communication networks with a shared communication medium, round-robin scheduling, cf.\ \cite{Katevenis1987}, is a starvation-free scheduling scheme in which every participant sends a message after its predecessor has sent its message. This process continues in a fixed order, without allowing one participant to send twice, until every participant has sent a message. Thereafter, the participant which has sent the first message is allowed to send a message again, then repeating the entire schedule in a cyclic fashion. By letting $x_{i}$ denote the messages of the $i$th participant, with $n$ overall participants in the network, round-robin scheduling is recovered from \eqref{eq:compression} by letting $c$ be the $n$-periodic sequence of vectors 
\begin{equation}
e_{1},\;e_{2},\;e_{3},\;\dots,\;e_{n},\;e_{1},\;e_{2},\;\dots \label{eq:standardbasis}
\end{equation}
with $e_{i}$ denoting the $i$th vector of the standard basis of $\mathbb{R}^{n}$, i.e. the vector whose $i$th entry is $1$ but all whose other entries are $0$. In particular, choosing $x$ and $c$ as above, $y$ will attain the message $x_{1}\left(0\right)$ of participant $1$ at time $0$, the message $x_{2}\left(1\right)$ of participant $2$ at time $1$, and so forth, until attaining the message $x_{n}\left(n-1\right)$ of participant $n$ at time $n-1$, thereafter returning to $x_{1}\left(n\right)$ again and repeating the process in a cyclic fashion. In this notation, the unsent messages, such as $x_{1}\left(1\right)$ or $x_{1}\left(2\right)$, are undisclosed to receivers.

\newpage

In digital photography, active pixel sensors, cf.\ \cite{Fossum1997}, have become a popular alternative to charge-coupled devices for their advantages in bloom, power consumption, lag, manufacturing, on-board image processing, scalability, and cost. Yet, these sensors are not capable of reading all pixel data simultaneously. Instead, pixels are read out line-by-line, which has been termed the rolling shutter effect. Letting $x_{i}$ denote the color intensities of the $i$th line of pixels, with $n$ overall lines of pixels, the rolling shutter effect is recovered from \eqref{eq:compression} by letting $c$ be the $n$-periodic sequence of vectors \eqref{eq:standardbasis}. Choosing $x$ and $c$ in such a fashion, \eqref{eq:compression} will yield the pixels $x_{1}\left(0\right)$ of line $1$ at time $0$, the pixels $x_{2}\left(1\right)$ of line $2$ at time $1$, and so forth, until obtaining the pixels of the $n$th line, $x_{n}\left(n-1\right)$ at time $n-1$, thereafter reading out the first line $x_{1}\left(n\right)$ again and repeating this process in a cyclic fashion.

In the present paper, we study conditions under which, in \eqref{eq:compression}, $x$ can be reconstructed from $y$. Thereby, the relation between $x$ and $y$ which is established will allow to explicitly compute $x$ from $y$. In the above applications, this would allow us to either reconstruct missed messages omitted due to round-robin scheduling or reconstruct images which were distorted due to rolling shutter. Our approach is based upon certain non-resonance conditions between the periodicities of $x$ and $c$. To our best knowledge, this approach is novel. However, other compression schemes bare similarities with \eqref{eq:compression}, which becomes particularly evident when depicting \eqref{eq:compression} as it is done in Fig.\ \ref{fig:genc}:

\begin{figure}
\centering
  \begin{tikzpicture}
\draw (3.5,-2) -- (5.5,-2) node[anchor=west,yshift=-1] {$y$};
 \draw (3.5,-2) -- (1,0) -- (-2,0) node[anchor=east] {$x_{1}$};
 \draw (3.5,-2) -- (1,-1) -- (-2,-1) node[anchor=east] {$x_{2}$};
 \draw (3.5,-2) -- (1,-2) -- (-2,-2) node[anchor=east] {$x_{3}$};
 \draw (3.5,-2) -- 	(1,-4) -- (-2,-4) node[anchor=east] {$x_{n}$};
    \draw (0,0) -- (0,.5) node[anchor=west] {$c_{1}$};
    \draw (0,-1) -- (0,-.5) node[anchor=west] {$c_{2}$};
    \draw (0,-2) -- (0,-1.5) node[anchor=west] {$c_{3}$};
    \draw (0,-4) -- (0,-3.5) node[anchor=west] {$c_{n}$};
  \draw[fill=white] (0,0) circle (.11) node {$\times$};
  \draw[fill=white] (0,-1) circle (.11) node {$\times$};
  \draw[fill=white] (0,-2) circle (.11) node {$\times$};
  \draw (0,-2.7) -- (0,-2.7) node{$\vdots$};
  \draw[fill=white] (0,-4) circle (.11) node {$\times$};
  \draw[fill=white] (3.5,-2) circle (.11) node {$+$};
\end{tikzpicture}
  \caption{The vector valued signal $x$ is compressed to a scalar-valued signal $y$ by taking inner products with the vector-valued mixing signal $c$.}\label{fig:genc}
\end{figure}
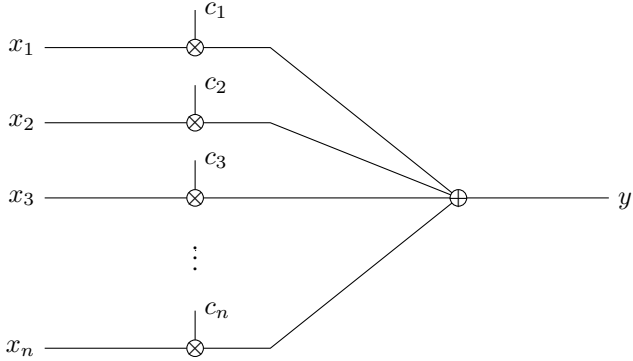

In quantization schemes for vector-valued signals, signals are either classified by distinct binary-valued classifiers, whose outputs are then summed up in order to produce a compressed output, cf.\ \cite[Figure 10.5]{Gersho1992}, or encoded by distinct codebooks, which are switched through, say periodically, in order to produce a compressed output, cf.\ \cite[Figure 14.3]{Gersho1992}.

In compression schemes for sparse signals, families of random vectors are sought such that the family of inner products of those vectors and the sparse signal allows to reconstruct the sparse signal, whereby the family of inner products can be viewed as compressed output, cf.\ \cite{Candes2006}.

In modulated wideband converters, a scalar-valued signal is multiplied with distinct mixing signals, say periodically, in order to produce distinct outputs, which are then sampled in order to produce compressed outputs, cf. \cite{Mishali2010}.

\section{Lossless Causal Compression \\ of Periodic Signals}

\begin{figure}
 \centering
 \begin{tikzpicture}
 \draw[white] (0,.525) -- (0,.525) node[anchor=west] {\phantom{$c_{1}$}};
 \draw (0,-1) -- (2.5,-2);
  \draw[dotted] (0,-2) -- (2.5,-2);	
  \draw[dotted] (0,0) -- (2.5,-2) node[pos=.5,anchor=south west] {$n$-periodic};	
  \draw[dotted] (0,-4) -- (2.5,-2);	
 \draw (0,0) -- (-2,0) node[anchor=east] {$x_{1}$};
 \draw (0,-1) -- (-2,-1) node[anchor=east] {$x_{2}$};
 \draw (0,-2) -- (-2,-2) node[anchor=east] {$x_{3}$};
 \draw (0,-4) -- (-2,-4) node[anchor=east] {$x_{n}$};
  \draw[fill=white] (0,0) circle (.1);
  \draw[fill=white] (0,-1) circle (.1);
  \draw[fill=white] (0,-2) circle (.1);
  \draw (0,-2.925) -- (0,-2.925) node{$\vdots$};
  \draw[fill=white] (0,-4) circle (.1);
  \draw (2.5,-2) -- (4.5,-2) node[anchor=west,yshift=-1] {$y$};
  \draw[fill=white] (2.5,-2) circle (.1);
   \draw[->] ([shift=(159:.8075)]2.5,-2) arc (159:180:.8075);
 \end{tikzpicture}
\caption{Letting $c$ be the periodic sequence of standard basis vectors \eqref{eq:standardbasis}, our compression scheme becomes an $n$-periodic switch.}\label{fig:percompress}
\end{figure}
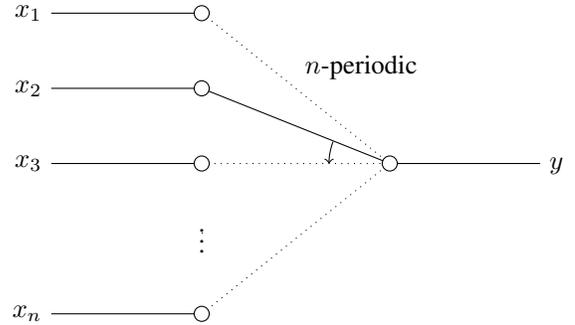

On the examples of communication networks with round-robin scheduling and digital photography with active pixel sensors, respectively, we outlined how the compression scheme \eqref{eq:compression}, with $c$ being the $n$-periodic sequence of standard basis vectors \eqref{eq:standardbasis}, implicitly arises in applications. This scheme lets $y$ switch through the signals $x_{i}$ periodically as if it was an $n$-periodic switch, illustrated in Fig.\ \ref{fig:percompress}. Motivated by these examples, we first focus on conditions for this scheme to be lossless, i.e., conditions under which $x$ can be uniquely reconstructed from $y$. To this end, let $x:\mathbb{N}_{0}\to\mathbb{R}^{n}$ be a $p$-periodic sequence, i.e.
\begin{equation}
 x\left(t+p\right)=x\left(t\right).
\end{equation}
We note that the sequence \eqref{eq:standardbasis} is $n$-periodic. Now $y$ attains the values $x_{1}\left(0\right)$, $x_{2}\left(1\right)$, and so forth, until attaining $x_{n}\left(n-1\right)$, thereafter returning to $x_{1}\left(n\right)$ again and repeating this process in a cyclic fashion. But as $x$ is $p$-periodic, the values of $x$ which were missed during the first cycle, e.g. $x_{1}\left(1\right)$, will occur again at the times $1+ip$, $i\in\mathbb{N}$. Thus, the question arises whether the periodic switch $c$ will let the signal $x_{1}$ pass at one of the times $1+ip$, $i\in\mathbb{N}$, i.e., whether there is some $j\in\mathbb{N}$ such that $1+ip=jn+1$. For knowing $x_{1}\left(2\right)$, we must thus ask for solvability of $2+ip=jn+1$ and so forth, and thus, in general, ask whether
\begin{equation}
\mathbb{N}\times\mathbb{N}\to\mathbb{Z},\;\;\;\left(i,j\right)\mapsto ip-jn  
\end{equation}
is onto. Rewriting $y\left(t\right)$ as $x_{\left(t\operatorname{mod}n\right)+1}\left(t\operatorname{mod}p\right)$, this question can be restated as follows: is it, for any $i$ and $j$, possible to find a time $t$ such that $t\operatorname{mod}n$ is $i$ and $t\operatorname{mod}p$ is $j$? In modular arithmetic, one poses this question by asking whether, for all $i$ and $j$, the simultaneous congruences
\begin{align}
 t&\equiv i\;\;\;\hspace{.75pt}\left(\operatorname{mod}n\right) \label{eq:cong1}\\
 t&\equiv j\;\;\;\left(\operatorname{mod}p\right) \label{eq:cong2}
\end{align}
can be solved for $t$. By the Chinese remainder theorem, the answer is in the affirmative when the moduli $p$ and $n$ are coprime. In other words, should $p$ and $n$ be coprime, then the compression \eqref{eq:compression}, with $c$ chosen as \eqref{eq:standardbasis}, is lossless. This is formalized in the following theorem.

\begin{theorem}\label{thm:main}
Let $c$ be the $n$-periodic sequence of standard basis vectors \eqref{eq:standardbasis} and let $x$ be $p$-periodic. If $p$ and $n$ are coprime, then $x$ can be uniquely reconstructed from $y$.
\end{theorem}
\begin{proof}
 Substituting the $n$-periodic sequence of standard basis vectors \eqref{eq:standardbasis} for $c$ in \eqref{eq:compression}, we obtain $y\left(t\right)=x_{\left(t\operatorname{mod}n\right)+1}\left(t\right)$. With $x$ being $p$-periodic, this equals $x_{\left(t\operatorname{mod}n\right)+1}\left(t\operatorname{mod}p\right)$, as we noted above. If, for all $i$ and $j$, the simultaneous congruences \eqref{eq:cong1}-\eqref{eq:cong2} can be solved for $t$, then $y$ has attained every value of every $x_{i}$ at least once, and thus $x$ can be reconstructed from $y$. By the Chinese remainder theorem, the claim is proven.
\end{proof}

With this theorem at hand, we are able to determine whether the compression \eqref{eq:compression}, with $c$ chosen as \eqref{eq:standardbasis}, is lossless, only based upon the periods $p$ and $n$. A possible advantage of this characterization is that it allows us to verify whether our compression is lossless even when $p$ is not known explicitly. In particular, for fixed $n$, all $p$ such that the greatest common divisor $\operatorname{gcd}\left(n,p\right)$ of $n$ and $p$ is $1$ allow to reconstruct $x$ from $y$. These $p$ can be computed explicitly, for instance via Euler's totient function.

The above characterization of lossless causal compressors has a quite geometric interpretation in terms of the function
\begin{equation}
 \mathbb{N}_{0}\to\mathbb{Z}/n\mathbb{Z}\times\mathbb{Z}/p\mathbb{Z},\;\;\;t\mapsto\left(t\operatorname{mod}n,t\operatorname{mod}p\right),
\end{equation}
whereby $\mathbb{Z}/n\mathbb{Z}$ are the integers modulo $n$ and thus $\mathbb{Z}/n\mathbb{Z}\times\mathbb{Z}/p\mathbb{Z}$ is a discrete torus. In particular, one asks whether this function is a winding of the discrete torus, i.e. whether it is surjective. With this point of view, coprimeness of $n$ and $p$ is interpreted as those two periods not being in resonance.  

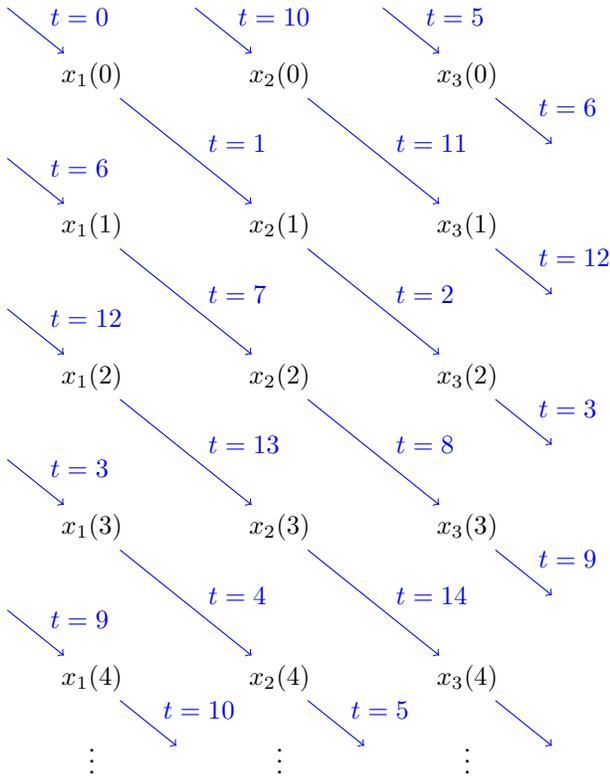
\begin{figure}[h]
\centering
 \begin{tikzpicture}
\def\dx{71pt}
\def\dy{57pt}
  \node (1) at (0*\dx,6*\dy) {$x_1(0)$};		\node (11) at (1*\dx,6*\dy) {$x_2(0)$}; 		\node (6) at (2*\dx,6*\dy) {$x_3(0)$};
\node (7) at (0*\dx,5*\dy) {$x_1(1)$};			\node (2) at (1*\dx,5*\dy) {$x_2(1)$}; 		\node (12) at (2*\dx,5*\dy) {$x_3(1)$};
\node (13) at (0*\dx,4*\dy) {$x_1(2)$};		\node (8) at (1*\dx,4*\dy) {$x_2(2)$}; 		\node (3) at (2*\dx,4*\dy) {$x_3(2)$};	
\node (4) at (0*\dx,3*\dy) {$x_1(3)$};			\node (14) at (1*\dx,3*\dy) {$x_2(3)$}; 		\node (9) at (2*\dx,3*\dy) {$x_3(3)$};		\node (31) at (2.5*\dx,3.5*\dy) {};
\node (10) at (0*\dx,2*\dy) {$x_1(4)$};		\node (5) at (1*\dx,2*\dy) {$x_2(4)$}; 		\node (15) at (2*\dx,2*\dy) {$x_3(4)$};
\node (19) at (0*\dx,1.5*\dy) {$\vdots$};		\node (20) at (1*\dx,1.5*\dy) {$\vdots$}; 		\node (21) at (2*\dx,1.5*\dy) {$\vdots$};

\node (41) at (1.5*\dx,1.5*\dy) {};
\node (51) at (1.5*\dx,6.5*\dy) {};
\node (61) at (2.5*\dx,5.5*\dy) {};
\node (71) at (-0.5*\dx,5.5*\dy) {};
\node (81) at (2.5*\dx,2.5*\dy) {};
\node (82) at (-0.5*\dx,2.5*\dy) {};
\node (83) at (0.5*\dx,1.5*\dy) {};
\node (84) at (0.5*\dx,6.5*\dy) {};
\node (85) at (2.5*\dx,4.5*\dy) {};
\node (86) at (-0.5*\dx,4.5*\dy) {};
\node (87) at (-0.5*\dx,3.5*\dy) {};
\node (88) at (-0.5*\dx,6.5*\dy) {};
\node (89) at (2.5*\dx,1.5*\dy) {};
 \draw [->,blue] (88) -- (1) node [pos=.6,anchor=south west] {$t=0$};
 \draw [->,blue] (1) -- (2) node [pos=.6,anchor=south west] {$t=1$};
 \draw [->,blue] (2) -- (3) node [pos=.6,anchor=south west] {$t=2$};
 \draw [->,blue] (3) -- (31) node [pos=.6,anchor=south west] {$t=3$};
\draw [->,blue] (87) -- (4) node [pos=.6,anchor=south west] {$t=3$};
 \draw [->,blue] (4) -- (5) node [pos=.6,anchor=south west] {$t=4$};
 \draw [->,blue] (5) -- (41) node [pos=.6,anchor=south west] {$t=5$};
 \draw [->,blue] (51) -- (6) node [pos=.6,anchor=south west] {$t=5$};
 \draw [->,blue] (6) -- (61) node [pos=.6,anchor=south west] {$t=6$};
 \draw [->,blue] (71) -- (7) node [pos=.6,anchor=south west] {$t=6$};
 \draw [->,blue] (7) -- (8) node [pos=.6,anchor=south west] {$t=7$};
 \draw [->,blue] (8) -- (9) node [pos=.6,anchor=south west] {$t=8$};
 \draw [->,blue] (9) -- (81) node [pos=.6,anchor=south west] {$t=9$};
 \draw [->,blue] (82) -- (10) node [pos=.6,anchor=south west] {$t=9$};
 \draw [->,blue] (10) -- (83) node [pos=.6,anchor=south west] {$t=10$};
\draw [->,blue] (84) -- (11) node [pos=.6,anchor=south west] {$t=10$};
\draw [->,blue] (11) -- (12) node [pos=.6,anchor=south west] {$t=11$};
\draw [->,blue] (12) -- (85) node [pos=.6,anchor=south west] {$t=12$};
\draw [->,blue] (86) -- (13) node [pos=.6,anchor=south west] {$t=12$};
\draw [->,blue] (13) -- (14) node [pos=.6,anchor=south west] {$t=13$};
\draw [->,blue] (14) -- (15) node [pos=.6,anchor=south west] {$t=14$};
\draw [->,blue] (15) -- (89);
    \end{tikzpicture}
    \caption{The graph of the function $t\mapsto\left(t\operatorname{mod}3,t\operatorname{mod}5\right)$ depicted as a winding of the discrete torus $\mathbb{Z}/3\mathbb{Z}\times\mathbb{Z}/5\mathbb{Z}$, thus attaining all its values.}
    \label{fig:discrete_torus}
  \end{figure}

  \newpage

  \begin{example}
   We consider $n=3$ signals $x_{1}$, $x_{2}$, $x_{3}$, which we assume to be $p=5$-periodic. When asking whether $x$ can be reconstructed from the compression scheme depicted in Fig.\ \ref{fig:percompress}, i.e., from the values 
   \begin{equation}
x_{1}\left(0\right),\;\;\;x_{2}\left(1\right),\;\;\;x_{3}\left(2\right),\;\;\;x_{1}\left(3\right),\;\;\;\dots,
   \end{equation}   
    we can employ Theorem \ref{thm:main} and compute the greatest common divisor of $3$ and $5$, which is $1$, thus revealing that $3$ and $5$ are coprime and that hence $x$ can be reconstructed from $y$. The geometric interpretation of this consequence is illustrated in Fig.\ \ref{fig:discrete_torus}. Therein, we see how the function $t\mapsto\left(t\operatorname{mod}3,t\operatorname{mod}5\right)$ is onto the discrete torus $\mathbb{Z}/3\mathbb{Z}\times\mathbb{Z}/5\mathbb{Z}$. For instance, $y\left(5\right)=x_{3}\left(5\right)=x_{3}\left(0\right)$, thus revealing the value of $x_{3}$ at time $0$, which was undisclosed to us during the first period. Proceeding, $y$ attains all values of of every $x_{i}$ at least once. Thus, here, our compression is lossless. On the other hand, if we consider the $n=2$ signals $x_{1}$, $x_{2}$ and assume them to be $p=4$-periodic, then $y$ will attain $x_{1}\left(0\right)$, $x_{2}\left(1\right)$, $x_{1}\left(2\right)$, and $x_{2}\left(3\right)$, thereafter returning to $x_{1}\left(4\right)=x_{1}\left(0\right)$ again and thus not generating new values of $x$ after the first period. For instance, $x_{1}\left(1\right)$ or $x_{1}\left(3\right)$ will never affect $y$. Here, the graph of $t\mapsto\left(t\operatorname{mod}2,t\operatorname{mod}4\right)$ will not cover the discrete torus $\mathbb{Z}/2\mathbb{Z}\times\mathbb{Z}/4\mathbb{Z}$, as depicted in Fig.\ \ref{fig:nowinding}, and thus, it will, in contrast to the previous setting, now not be possible to reconstruct $x$ from $y$. 
  \end{example}

  \begin{figure}[t]
  \centering
 \begin{tikzpicture}
\def\dx{100pt}
\def\dy{40pt}
  \node (10) at (0*\dx,6*\dy) {$x_1(0)$};		\node (20) at (1*\dx,6*\dy) {$x_2(0)$}; 
\node (11) at (0*\dx,5*\dy) {$x_1(1)$};			\node (21) at (1*\dx,5*\dy) {$x_2(1)$}; 
\node (12) at (0*\dx,4*\dy) {$x_1(2)$};		\node (22) at (1*\dx,4*\dy) {$x_2(2)$};
\node (13) at (0*\dx,3*\dy) {$x_1(3)$};			\node (23) at (1*\dx,3*\dy) {$x_2(3)$}; 	
\node (14) at (0*\dx,2.25*\dy) {$\vdots$};		\node (24) at (1*\dx,2.25*\dy) {$\vdots$}; 	
\node (00) at (-0.5*\dx,6.5*\dy) {};
\node (02) at (-0.5*\dx,4.5*\dy) {};
\node (31) at (1.5*\dx,4.5*\dy) {};
\node (32) at (1.5*\dx,2.5*\dy) {};
 \draw [->,blue] (00) -- (10) node [pos=.5,anchor=south west] {$t=4$};
  \draw [->,blue] (10) -- (21) node [pos=.5,anchor=south west] {$t=1$};
  \draw [->,blue] (02) -- (12) node [pos=.5,anchor=south west] {$t=2$};
  \draw [->,blue] (21) -- (31) node [pos=.5,anchor=south west] {$t=2$};
  \draw [->,blue] (12) -- (23) node [pos=.5,anchor=south west] {$t=3$};  
  \draw [->,blue] (23) -- (32) node [pos=.5,anchor=south west] {$t=4$};
    \end{tikzpicture}
    \caption{The graph of the function $t\mapsto\left(t\operatorname{mod}2,t\operatorname{mod}4\right)$ is not onto the discrete torus $\mathbb{Z}/2\mathbb{Z}\times\mathbb{Z}/4\mathbb{Z}$, for instance $x_{1}\left(1\right)$ or $x_{1}\left(3\right)$ are never attained.}
    \label{fig:nowinding}
  \end{figure}
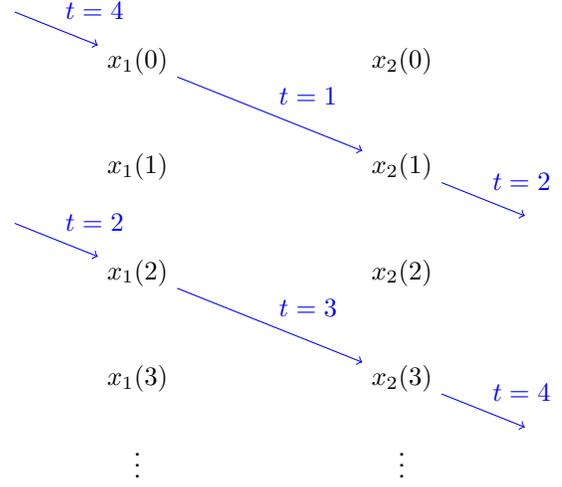

Hitherto, we have assumed $c$ to take the form \eqref{eq:standardbasis}. Next, we will turn our attention to more general $c$, thus returning to the general scheme depicted in Fig.\ \ref{fig:percompress}. Thereby, we will employ the insights gained from considering the greatest common divisor of $n$ and $p$. Now, assume $c$ to be $m$-periodic, 
\begin{equation}
 c\left(t+m\right)=c\left(t\right),\;\;\;m\geq n,
\end{equation}
with the property that any $n$ vectors chosen among $c\left(0\right)$, $\dots$, $c\left(m\right)$ will span $\mathbb{R}^{n}$. Before, this was the case with $m=n$. This lets the compression \eqref{eq:compression} be lossless provided that $m$ is greater than or equal to $n\operatorname{gcd}\left(m,p\right)$. If $m$ is, on the other hand, smaller than $n\operatorname{gcd}\left(m,p\right)$, then the compression is not lossless. This is stated in the following theorem.

\begin{theorem}
 Let $c$ be $m$-periodic with $m\geq n$ and let $x$ be $p$-periodic. If $c$ has the property that any $n$ vectors chosen from $c\left(0\right)$, $\dots$, $c\left(m\right)$ span $\mathbb{R}^{n}$, then $x$ can be reconstructed from $y$ if and only if $m\geq n\operatorname{gcd}\left(m,p\right)$.
\end{theorem}

\begin{proof}
Some fixed $x\left(t\right)$ can be reconstructed from $y$ if and only if the vectors $c\left(t+ip\right)$, $i\in\mathbb{N}_{0}$, span $\mathbb{R}^{n}$. But with $c$ being $m$-periodic, it is sufficient to ask whether $c\left(\left(t+ip\right)\operatorname{mod}m\right)$, $i\in\mathbb{N}_{0}$, spans $\mathbb{R}^{n}$. Since it was assumed that any pair of $n$ vectors chosen from $c\left(0\right)$, $\dots$, $c\left(m\right)$ spans $\mathbb{R}^{n}$, we must only ask whether $i\mapsto\left(t+ip\right)\operatorname{mod}m$ assumes at least $n$ distinct values. With $i\mapsto\left(t+ip\right)\operatorname{mod}m$ being $m/\operatorname{gcd}\left(m,p\right)$-periodic, for $c\left(t+ip\right)$, $i\in\mathbb{N}_{0}$, to span $\mathbb{R}^{n}$, it is necessary and sufficient for that period to be greater than or equal to $n$. This was claimed.
\end{proof}

When contrasting the conditions from the previous theorem with the conditions from Theorem \ref{thm:main}, where $m$ was $n$, we find that $\operatorname{gcd}\left(n,p\right)$ was previously $1$, and that hence the inequality $m\geq n\operatorname{gcd}\left(m,p\right)$ was, before, met with equality. The foregoing result thus extends Theorem \ref{thm:main} in two directions: On the one hand, it generalizes the sequences which can be taken into account for $c$ from \eqref{eq:standardbasis} to arbitrary periodic $c$ for which $n$ vectors taken from one period are linearly independent. On the other hand, it is not assumed that $m$ and $p$ are coprime, but the greatest common divisor of $m$ and $p$ can be an arbitrary number smaller than or equal to $m/n$, i.e.
\begin{equation}
 \operatorname{gcd}\left(m,p\right)\leq m/n.
\end{equation}
In other words, here, the greatest common divisor of $m$ and $p$ must be sufficiently small, but not necessarily $1$.
  
We employed the condition that any $n$ vectors taken from one period of $c$ are linearly independent. This condition was meant to ensure that $c\left(\left(t+ip\right)\operatorname{mod}m\right)$, $i\in\mathbb{N}_{0}$, spans $\mathbb{R}^{n}$. An alternative formulation of this condition is that any $n\times n$ matrix $C$ whose columns are chosen from $c\left(0\right)$, $\dots$, $c\left(m\right)$, has rank $n$. This, in turn, is equivalent to the Gramian $CC^{\top}$ having rank $n$ and hence being positive definite. Summarizing, the Gramian  
\begin{equation}
\sum_{t\in T}c\left(t\right)c\left(t\right)^{\top}=CC^{\top}
\end{equation}
must be positive definite for any indexing set $T$ containing $n$ distinct integers chosen from $0$, $\dots$, $m$. This ensures that the signal $c$ recurrently points in sufficiently many distinct directions of $\mathbb{R}^{n}$. Such a signal $c$ is said to be sufficiently rich. It shall be remarked that this condition resembles the persistence of excitation condition known from adaptive filtering \cite{Bitmead1984}. The distinction between our condition and persistence of excitation of $c$ is that persistence of excitation only asks for $n$ consecutive values of $c$ to span $\mathbb{R}^{n}$, while our condition requires that any $n$ values of $c$ taken from one period span $\mathbb{R}^{n}$. Thus, our condition implies persistence of excitation, but not vice versa.
  
Above, we have devoted our attention to arbitrary periodic signals $x$. In applications, yet, it will often be the case that the signal $x$ is produced by some dynamical system. Should this be the case, then properties of that dynamical system can be taken into account in order to provide tangible conditions which are more particular than the generic conditions derived above. This will be the scope of the following section.
  
\section{Compression of Signals \\ Produced by Exosystems}
In the previous section, we asked whether the compressed signal $y$, defined by \eqref{eq:compression}, allows to reconstruct $x$ uniquely. Thereby, $x$ was an arbitrary signal which we assumed to be periodic. In the present section, we will let the evolution of $x$ be determined by a dynamical system, which is in this context referred to as an exosystem or signal generator. First, we will assume that all signals evolve in discrete time, as it was also presumed in the foregoing section. Thereafter, we will also consider signals evolving in discrete time.

\subsection{Exosystems Evolving in Discrete Time}
In this subsection, we consider signals $x$ which are generated by difference equations, i.e., by exosystems evolving in discrete time. First, recalling our motivation from image processing pointed our in the introduction, we will consider systems evolving under permutations, say a time series of pictures in which the pixels permute as time proceeds. This is modeled by letting each signal $x_{i}$ evolve under
\begin{equation}
 x_{i}\left(t+1\right)=x_{\sigma\left(i\right)}\left(t\right) \label{eq:permut}
\end{equation}
where $\sigma$ is a permutation of the indices $\lbrace 1,\dots,n\rbrace$. We now, again, let $c$ be the $n$-periodic sequence of standard basis vectors \eqref{eq:standardbasis}, motivated by a pixel readout subject to rolling shutter, as it was noted in the introduction. Under these circumstances, surjectivity of 
\begin{equation}
 \mathbb{N}_{0}\times\lbrace1,\dots,n\rbrace\to\lbrace 1,\dots,n\rbrace,\;\;\;\left(i,j\right)\mapsto\sigma^{in+j-1}\left(j\right) \label{eq:iterpermut}
\end{equation}
is necessary and sufficient for lossless compression. This is stated in the following proposition.
\begin{proposition}
Let $c$ be the $n$-periodic sequence of standard basis vectors \eqref{eq:standardbasis} and let all $x_{i}$ evolve under \eqref{eq:permut}, for some permutation $\sigma$. Then $x$ can be uniquely reconstructed from $y$ if and only if \eqref{eq:iterpermut} is surjective.
\end{proposition}
\begin{proof}
Consider some fixed $t$. Euclidean division by $n$ yields 
\begin{equation}
 t=in+j-1 \label{eq:Euclid}
\end{equation}
where $i$ is the quotient and $j-1=t\operatorname{mod}n$ is the remainder. With this notation and $c$ being the $n$-periodic sequence of standard basis vectors \eqref{eq:standardbasis}, we find that $y\left(t\right)=x_{j}\left(in+j-1\right)$. Repeating the iteration \eqref{eq:permut} for $in+j-1$ times leaves us with $y\left(t\right)=x_{\sigma^{in+j-1}\left(j\right)}\left(0\right)$. Thus $x\left(0\right)$ can be uniquely reconstructed if and only if $t\mapsto\sigma^{in+j-1}\left(j\right)$, with $i$ and $j$ related to $t$ via Euclidean division, attains every index at least once. All other $x\left(t\right)$ are reconstructed by applying $\sigma$.
\end{proof}

Surjectivity of \eqref{eq:iterpermut} has a number-theoretical interpretation in terms of cycles of $\sigma$. If $n$ is a common multiple of all cycle lengths $\ell_{1},\dots,\ell_{m}$ of $\sigma$, i.e., if there is some $p\in\mathbb{N}$ such that
\begin{equation}
\operatorname{lcm}\left(\ell_{1},\dots,\ell_{m}\right)=n/p,
\end{equation}
then the image of \eqref{eq:iterpermut} cannot be greater than the image of $j\mapsto\sigma^{j-1}\left(j\right)$, i.e., for the image of \eqref{eq:iterpermut} to be large and thus reveal new information about $x$, the cycle lengths of $\sigma$ must not be in resonance with $n$.

\newpage

\begin{example}
We consider $n=5$ signals evolving under the iteration \eqref{eq:permut}, where $\sigma$, in cycle notation, reads
\begin{equation}
 \left(4\;2\;3\;1\right)\;\left(5\right).
\end{equation}
Here, the iterated permutation \eqref{eq:iterpermut} attains every index. Although $j\mapsto\sigma^{j-1}\left(j\right)$ attains the values $1$, $3$, $4$, and $5$, it does not attain the value $2$, i.e., $x_{2}\left(0\right)$ cannot be reconstructed before $t=n=5$. But as $n=5$ is no common multiple of the cycle lengths $\ell_{1}=4$ and $\ell_{2}=1$, the compressed signal $y$ will provide new information after $t=5$. In particular, Euclidean division of $t=7$ by $n=5$ yields $1n+3-1$ and $j=3$ is mapped to $2$ by $\sigma^{7}$, i.e., $x_{2}\left(0\right)$ becomes visible in $y\left(7\right)$. If we were to delete the $1$-cycle $5\mapsto 5$ to only consider $n=4$ signals, then this effect would no longer occur and hence $x$ could not be reconstructed from $y$. The reason for this is that now the only left cycle length $\ell_{1}=4$ equals $n=4$ and that hence \eqref{eq:iterpermut} does not cover more indices than $j\mapsto\sigma^{j-1}\left(j\right)$, which, as we saw above, never attains $2$, leaving $x_{2}\left(0\right)$ undisclosed to us.
\end{example}

Above we saw how resonances can again, also for signals produced by exosystems, help us to understand lossless compressors. This shall be reason enough to consider $2$ signals evolving under repeated application of a rotation. This, at first glance, appears to be very particular, but will extend to more general exosystems momentarily. For now, consider 
\begin{align}
 x_{1}\left(t+1\right)&=x_{1}\left(t\right)\operatorname{cos}\left(\alpha\right)-x_{2}\left(t\right)\hspace{.5pt}\operatorname{sin}\hspace{.5pt}\left(\alpha\right), \label{eq:rotsys1}\\
 x_{2}\left(t+1\right)&=x_{1}\left(t\right)\hspace{.5pt}\operatorname{sin}\hspace{.5pt}\left(\alpha\right)+x_{2}\left(t\right)\operatorname{cos}\left(\alpha\right), \label{eq:rotsys2}
\end{align}
i.e., the vector $x\left(t\right)$ is rotated counter-clockwise through the angle $\alpha\in\left(0,2\pi\right)$ in order to produce $x\left(t+1\right)$. Again choosing $c$ as \eqref{eq:standardbasis}, a sufficient condition for lossless compression is that it is possible to find an odd and an even multiple of $\alpha$ which are the same angle, i.e., which yield the same real number modulo $2\pi$. More formally,
\begin{equation}
 \exists\hspace{.75pt} p,q\in\mathbb{N}_{0}:\;\;\;\left(2p+1\right)\alpha\equiv2q\alpha\;\;\;\hspace{.75pt}\left(\operatorname{mod}2\pi\right) \label{eq:rotationalresonance}
\end{equation}
leads to lossless compression of signals produced by the exosystem \eqref{eq:rotsys1}-\eqref{eq:rotsys2}, as stated in the following proposition.

\begin{proposition}\label{prop:rot}
 Let $c$ be the $n$-periodic sequence of standard basis vectors \eqref{eq:standardbasis} and let $x$ evolve under \eqref{eq:rotsys1}-\eqref{eq:rotsys2}, for some $\alpha\in\left(0,2\pi\right)$. If there is an odd and an even multiple of $\alpha$ which are the same real number modulo $2\pi$, then $x$ can be uniquely reconstructed from $y$.
\end{proposition}
\begin{proof}
For odd $t$, we notice that $y\left(t\right)$ equals $x_{1}\left(t\right)$ and for even $t$, one finds that $y\left(t\right)$ equals $x_{2}\left(t\right)$. Applying \eqref{eq:rotsys1}-\eqref{eq:rotsys2} for $t$ iterations, we find that $y\left(t\right)$ is given by 
\begin{equation}
 y\left(t\right)=\begin{cases}x_{1}\left(0\right)\hspace{.5pt}\operatorname{sin}\hspace{.5pt}\left(\alpha t\right)+x_{2}\left(0\right)\operatorname{cos}\left(\alpha t\right) & t\text{ odd} \\ x_{1}\left(0\right)\operatorname{cos}\left(\alpha t\right)-x_{2}\left(0\right)\hspace{.5pt}\operatorname{sin}\hspace{.5pt}\left(\alpha t\right) & t\text{ even}\end{cases}
\end{equation}
such that, under the condition \eqref{eq:rotationalresonance}, the real numbers $y\left(2q\right)$ and $y\left(2p+1\right)$ are related to the vector $x\left(0\right)$ through a rotation matrix which rotates $x\left(0\right)$ counter-clockwise through the angle $\left(2p+1\right)\alpha\operatorname{mod}2\pi=2q\alpha\operatorname{mod}2\pi$. But rotation matrices are invertible. Thus $x\left(0\right)$ can be uniquely reconstructed from $y\left(2q\right)$ and $y\left(2p+1\right)$. All other $x\left(t\right)$ are reconstructed by counter-clockwise rotation through $\alpha$.
\end{proof}

 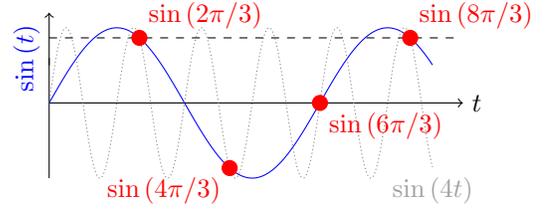
\begin{figure}
 \centering
 
    \begin{tikzpicture}
\draw[->] (0,-1) -- (0,1.2);
\draw (0,.5) -- (0,.6)  node[anchor=south,rotate=90] {\color{blue}$\operatorname{sin}\left(t\right)$};
\draw[->] (0,0) -- (5.5,0) node[anchor=west] {$t$};
   \draw[domain=0:5.1,samples=1000,smooth,variable=\x,blue] plot ({\x},{sin(100*\x )});
   \draw[domain=0:5.1,samples=1000,smooth,variable=\x,gray,densely dotted] plot ({\x},{sin(400*\x )});
   \draw[white] (5.1,{sin(2*100*3.6/3)}) -- (5.1,{sin(2*100*3.6/3)}) node[anchor=north] {\color{gray!75!white}$\operatorname{sin}\left(4t\right)$};
\draw[dashed] (0,{sin(4*100*3.6/3)}) -- (5.4,{sin(4*100*3.6/3)});
   \filldraw[red] ({3.6/3},{sin(100*3.6/3)}) circle (.1) node[anchor=south west] {$\operatorname{sin}\left(2\pi/3\right)$};
\filldraw[red] ({2*3.6/3},{sin(2*100*3.6/3)}) circle (.1) node[anchor=north east] {$\operatorname{sin}\left(4\pi/3\right)$};
\filldraw[red] ({3*3.6/3},{sin(3*100*3.6/3)}) circle (.1) node[anchor=north west] {$\operatorname{sin}\left(6\pi/3\right)$};
\filldraw[red] ({4*3.6/3},{sin(4*100*3.6/3)}) circle (.1) node[anchor=south west] {$\operatorname{sin}\left(8\pi/3\right)$};
\end{tikzpicture}

\caption{The odd and even multiples $2\pi/3$ and $8\pi/3$ of $2\pi/3$ are the same real numbers modulo $2\pi$, so $\operatorname{sin}\left(t\right)$ assumes the same value at $t=2\pi/3$ and at $t=8\pi/3$, but also $\operatorname{sin}\left(t\right)=\operatorname{sin}\left(4t\right)$ when $t$ is a multiple of $2\pi/3$.}\label{fig:res}
 \end{figure}

\begin{example}
Consider $n=2$ signals evolving under \eqref{eq:rotsys1}-\eqref{eq:rotsys2} where $\alpha$ is taken to be $2\pi/3$. It is now the case that $\alpha$ equals $4\alpha$ modulo $2\pi$, i.e., \eqref{eq:rotationalresonance} holds. In particular, $x$ can be reconstructed from $y\left(1\right)$ and $y\left(4\right)$. Two different interpretations of the resonance 
\begin{equation}
2\pi/3\equiv 8\pi/3\;\;\; \left(\operatorname{mod}2\pi\right)
\end{equation}
are depicted in Fig. \ref{fig:res}. Therein, we find that, on the one hand, $\operatorname{sin}\left(2\pi/3\right)$ equals $\operatorname{sin}\left(8\pi/3\right)$, and on the other hand, $\operatorname{sin}\left(t\right)$ equals $\operatorname{sin}\left(4t\right)$ at all $t$ which are integer multiples of $2\pi/3$. If we took $\alpha=\pi$, then, although the condition \eqref{eq:rotationalresonance} now cannot be satisfied, for any even multiple of $\pi$ equals zero modulo $2\pi$ and every odd multiple of $\pi$ equals $\pi$ modulo $2\pi$, still the matrix relating $y\left(1\right)$ and $y\left(2\right)$ to $x\left(0\right)$ has full rank and thus $x$ can be uniquely reconstructed from $y$. This illustrates that \eqref{eq:rotationalresonance} is sufficient, but not necessary for lossless compression.
\end{example}

As pointed out beforehand, \eqref{eq:rotsys1}-\eqref{eq:rotsys2} should only serve as a prototype for a more general class of exosystems, viz. for
\begin{equation}
 x\left(t+1\right)=Gx\left(t\right), \label{eq:groupsys}
\end{equation}
where $G$ is a real $n\times n$ matrix that belongs to some matrix (Lie) group $\operatorname{G}$, e.g. the rotation group $\operatorname{SO}\left(2\right)$ as above, the orthogonal group $\operatorname{O}\left(n\right)$, or the special linear group $\operatorname{SL}\left(n\right)$. We recall that permutation matrices are also a group, so \eqref{eq:iterpermut} falls under the above class of exosystems as well. During the proof of Proposition \ref{prop:rot}, the crucial finding was that a rotation matrix related some particular $y\left(t\right)$ to $x\left(0\right)$. As rotation matrices are invertible, this insight allowed us to conclude that our compression is lossless. Now, that we assumed $x$ to evolve under repeated left-multiplication with some $G\in\operatorname{G}$, we must hope that some other matrix from $\operatorname{G}$, call it $G^{\prime}$, will relate some particular $y\left(t\right)$ to $x\left(0\right)$. As all matrices in $\operatorname{G}$ are invertible, this would again allow us to uniquely reconstruct $x$. As we will see next, if $G^{\prime}\in\operatorname{G}$ satisfies
\begin{equation}
 \forall j=1,\dots,n\;\;\;\;\;\exists i\in\mathbb{N}_{0}:\;\;\;\;\; G^{in+j-1}=G^{\prime}, \label{eq:groupresonance}
\end{equation}
then the aforementioned relation between $y$ and $x$ will hold and thus compression of signals produced by \eqref{eq:groupsys} through the sequence of vectors \eqref{eq:standardbasis} is lossless. This is also claimed in the following proposition.
\begin{proposition}
 Let $c$ be the $n$-periodic sequence of standard basis vectors \eqref{eq:standardbasis} and let $x$ evolve under \eqref{eq:groupsys}, for some $G\in\operatorname{G}$. If, there is an $G^{\prime}\in\operatorname{G}$ with the property that, for all $j=1,\dots,n$, it is possible to find an $i\in\mathbb{N}_{0}$ such that $G^{in+j-1}=G^{\prime}$, then $x$ can be uniquely reconstructed from $y$.
 \end{proposition}

 \begin{proof}
Consider some fixed $t$. Euclidean division by $n$ yields \eqref{eq:Euclid} where $i$ is the quotient and $j-1=t\operatorname{mod}n$ is the remainder. With this notation and $c$ being the $n$-periodic sequence of standard basis vectors \eqref{eq:standardbasis}, we find that $y\left(t\right)=x_{j}\left(in+j-1\right)$. Repeating the iteration \eqref{eq:groupsys} for $in+j-1$ times reveals that $y\left(t\right)$ is the inner product of $x_{0}$ and the $j$th row of $G^{in+j-1}$. So if there is an $G^{\prime}$ in $\operatorname{G}$ with the property \eqref{eq:groupresonance}, then, for every row of $G^{\prime}$, it will be possible to find $t$ such that $y\left(t\right)$ is the inner product of $x_{0}$ and that row. Thus there will be some $y\left(t\right)$ which are related to $x_{0}$ through $G^{\prime}$. As $G^{\prime}$ is from $\operatorname{G}$, it must be invertible. Therefore, $x\left(0\right)$ can be uniquely reconstructed from those $y\left(t\right)$. All other $x\left(t\right)$ are reconstructed by left-multiplication with $G$.
 \end{proof}

 \begin{example}
Consider $n=3$ signals evolving under the iteration \eqref{eq:groupsys} where $G$ is taken to be the real $3\times 3$ matrix
 \begin{equation}
  G=\begin{pmatrix}-1 & 1 & \hphantom{-}1 \\ \hphantom{-}1 & 1 & -1 \\ -3/2 & 3/2 & \hphantom{-}1\end{pmatrix}
 \end{equation}
which is a member of the special linear group $\operatorname{G}=\operatorname{SL}\left(3\right)$. We find that the powers $G^{3}$, $G^{7}$, and $G^{11}$ all equal 
\begin{equation}
 G^{\prime}=\begin{pmatrix}5/2 & 1/2 & -2 \\ 1/2 & 1/2 & \hphantom{-}0 \\ 3 & 0 & -2\end{pmatrix}
\end{equation}
which also belongs to $\operatorname{SL}\left(3\right)$. Euclidean division of $3$, $7$, and $11$ by $n=3$ yields the remainders $0$, $1$, and $2$, so \eqref{eq:groupresonance} holds and thus our compression is lossless. In particular, $x$ can be reconstructed from $y\left(3\right)$, $y\left(7\right)$, and $y\left(11\right)$.
\end{example}

Approaching the end of this subsection, we are now in the position to return to one of our motivating examples: communication networks with a shared communication medium and round-robin scheduling. We assume that measurements from $N$ processes of the type \eqref{eq:rotsys1}-\eqref{eq:rotsys2} are taken, i.e., we aim to reconstruct the evolution of the dynamical systems 
\begin{equation}
 x_{i}\left(t+1\right)=\begin{pmatrix}\operatorname{cos}\left(\alpha_{i}\right) & -\operatorname{sin}\left(\alpha_{i}\right) \\ \operatorname{sin}\left(\alpha_{i}\right) & \hphantom{-}\operatorname{cos}\left(\alpha_{i}\right)\end{pmatrix}x_{i}\left(t\right) \label{eq:rrsyss}
\end{equation}
with $i$ ranging from $1$ to $N$. Further, let every system be equipped with a sensing device returning the measurements $y_{i}\left(t\right)=\langle e_{1},x_{i}\left(t\right)\rangle$, i.e., the first entry of the $i$th signal. But those measurements are transmitted over a shared communication medium with round-robin scheduling, as it is depicted in Fig.\ \ref{fig:rr}, i.e., at time $t$, the central processing entity which gathers all measurements will only obtain $y\left(0\right)=y_{1}\left(0\right)$, $y\left(1\right)=y_{2}\left(1\right)$, and so forth, i.e., $y\left(t\right)=y_{\left(t\operatorname{mod}N\right)+1}\left(t\right)$. This is also modeled by stacking all the signals $x_{i}$ into a $\mathbb{R}^{2N}$-valued signal $x$ and letting $c$ be the sequence of odd-indexed standard basis vectors 
\begin{equation}
 e_{1},\;e_{3},\;e_{5},\;\dots,\;e_{2N-1},\;e_{1},\;e_{3},\;\dots \label{eq:oddvectors}
\end{equation}
which is $N$-periodic. Our objective is to reconstruct the signals $x_{i}$ produced by the systems \eqref{eq:rrsyss} although only knowing the messages $y$ which were transmitted over the shared communication medium with round-robin scheduling. 

We will show that reconstructing $x$ is possible if no $N$-fold multiple of any $\alpha_{i}$ is a multiple of $2\pi$, i.e., if it holds true that
\begin{equation}
\forall i=1,\dots,N,\;\;\;\;\;\;\; N\alpha_{i}\not\equiv 0\;\;\;\left(\operatorname{mod}2\pi\right). \label{eq:Nalphaineq0mod2pi}
\end{equation}

\begin{figure}
 \centering
 \begin{tikzpicture}
  \begin{scope}[xscale=-.6,shift={(-7.5,2.75)}]
 \draw (0,-1) -- (2.5,-2);
  \draw[dotted] (0,-2) -- (2.5,-2);	
  \draw[dotted] (0,0) -- (2.5,-2) node[pos=.6,anchor=south east] {$n$-periodic};	
  \draw[dotted] (0,-4) -- (2.5,-2);	
  \draw[->] ([shift=(159:.8075)]2.5,-2) arc (159:180:.8075);  
  \end{scope}
\draw[rounded corners=3pt] (.05,.05) rectangle (1.95,1.45);
 \draw (0,0) rectangle (2,1.5) node [pos=.5, align=center] {central \\ processing \\ entity};
 \begin{scope}[shift={(.5,-.5)}]
 \draw[fill=black] (7,1.25) circle (.075);
 \draw[fill=white] (5,1) rectangle (7,1.5) node [pos=.5] {sensor 3};
 \draw[line width=1.5pt] (7,.993) -- (7,1.508);
 \draw (7.15,1.25) arc (0:70:.1);
 \draw (7.15,1.25) arc (0:-70:.1);
 \draw (7.2,1.25) arc (0:70:.125);
 \draw (7.2,1.25) arc (0:-70:.125);
 \draw (7.25,1.25) arc (0:70:.15);
 \draw (7.25,1.25) arc (0:-70:.15);
 \draw (4,1.25) -- (5,1.25);
 \draw[fill=white] (4,1.25) circle (.1);
 \end{scope}
  \begin{scope}[shift={(.5,.5)}]
 \draw[fill=black] (7,1.25) circle (.075);
 \draw[fill=white] (5,1) rectangle (7,1.5) node [pos=.5] {sensor 2};
 \draw[line width=1.5pt] (7,.993) -- (7,1.508);
 \draw (7.15,1.25) arc (0:70:.1);
 \draw (7.15,1.25) arc (0:-70:.1);
 \draw (7.2,1.25) arc (0:70:.125);
 \draw (7.2,1.25) arc (0:-70:.125);
 \draw (7.25,1.25) arc (0:70:.15);
 \draw (7.25,1.25) arc (0:-70:.15);
 \draw (4,1.25) -- (5,1.25);
 \draw[fill=white] (4,1.25) circle (.1);
 \end{scope}
 \begin{scope}[shift={(.5,1.5)}]
 \draw[fill=black] (7,1.25) circle (.075);
 \draw[fill=white] (5,1) rectangle (7,1.5) node [pos=.5] {sensor 1};
 \draw[line width=1.5pt] (7,.993) -- (7,1.508);
 \draw (7.15,1.25) arc (0:70:.1);
 \draw (7.15,1.25) arc (0:-70:.1);
 \draw (7.2,1.25) arc (0:70:.125);
 \draw (7.2,1.25) arc (0:-70:.125);
 \draw (7.25,1.25) arc (0:70:.15);
 \draw (7.25,1.25) arc (0:-70:.15);
 \draw (4,1.25) -- (5,1.25);
 \draw[fill=white] (4,1.25) circle (.1);
 \end{scope}
 \begin{scope}[shift={(.5,-2.5)}]
 \draw[fill=black] (7,1.25) circle (.075);
 \draw[fill=white] (5,1) rectangle (7,1.5) node [pos=.5] {sensor $N$};
 \draw[line width=1.5pt] (7,.993) -- (7,1.508);
 \draw (7.15,1.25) arc (0:70:.1);
 \draw (7.15,1.25) arc (0:-70:.1);
 \draw (7.2,1.25) arc (0:70:.125);
 \draw (7.2,1.25) arc (0:-70:.125);
 \draw (7.25,1.25) arc (0:70:.15);
 \draw (7.25,1.25) arc (0:-70:.15);
 \draw (4,1.25) -- (5,1.25);
 \draw[fill=white] (4,1.25) circle (.1);
 \end{scope}
 \draw[white] (4.5,-.25) -- (4.5,-.25) node[yshift=3] {\color{black}$\vdots$};
 \draw (3,.75) -- (2,.75);
 \draw[fill=white] (3,.75) circle (.1);
\end{tikzpicture}
\caption{A group of $N$ sensing devices transmitting their measurements over a shared communication medium with round-robin scheduling.}\label{fig:rr}
\end{figure}
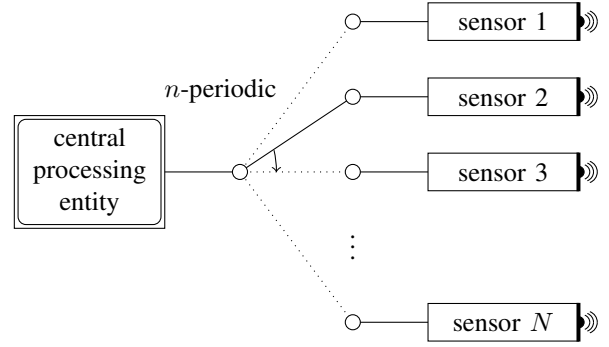

\begin{proposition}
Let $c$ be the $N$-periodic sequence of standard basis vectors \eqref{eq:oddvectors} and let all $x_{i}$ evolve under \eqref{eq:rrsyss}, for some $\alpha_{i}\in\left(0,2\pi\right)$. Then $x$ can be uniquely reconstructed from $y$ if and only if no $N$-fold multiple of any $\alpha_{i}$ is a multiple of $2\pi$.
\end{proposition}
\begin{proof}
Consider some fixed $t$. Euclidean division by $N$ yields $t=jN+i-1$ where $j$ is the quotient and $i-1=t\operatorname{mod}N$ is the remainder. With this notation and $c$ being the $N$-periodic sequence of standard basis vectors \eqref{eq:oddvectors}, we find that 
\begin{align}
y\left(t\right)=\left\langle\begin{pmatrix}\hphantom{-}\operatorname{cos}\left(\alpha_{i}t\right) \\ -\operatorname{sin}\left(\alpha_{i}t\right)\end{pmatrix},x_{i}\left(0\right)\right\rangle.  
\end{align}
In order to proceed, it remains to notice that the vectors 
\begin{equation}
\begin{pmatrix}\hphantom{-}\operatorname{cos}\left(\left(i-1\right)\alpha_{i}+jN\alpha_{i}\right) \\ -\operatorname{sin}\left(\left(i-1\right)\alpha_{i}+jN\alpha_{i}\right)\end{pmatrix},\;\;\;j\in\mathbb{N}_{0}
\end{equation}
span $\mathbb{R}^{2}$ unless $N\alpha_{i}$ is a multiple of $2\pi$. Thus, if $N\alpha_{i}$ is no multiple of $2\pi$, then $x_{i}\left(0\right)$ can be uniquely reconstructed from $y$. All other $x_{i}\left(t\right)$ are reconstructed by counter-clockwise rotation through $\alpha_{i}$. If no $N\alpha_{i}$ is a multiple of $2\pi$, then all $x_{i}$ can be reconstructed in this fashion.
\end{proof}

Our previous finding can be cast more generically for the case that the signals $x_{i}$ are produced by $x_{i}\left(t+1\right)=A_{i}x\left(t\right)$, $A_{i}$ some real $n\times n$ matrix, $i$ ranging from $1$ to $N$, and the measurements $y_{i}\left(t\right)=C_{i}x_{i}\left(t\right)$ are taken but switched through periodically due to round-robin scheduling, again as depicted in Fig.\ \ref{fig:rr}. In that case, for being able to reconstruct all $x_{i}$ despite round-robin scheduling, the matrices
\begin{equation}
 \begin{pmatrix}C_{i} \\ C_{i}A_{i}^{N} \\ C_{i}A_{i}^{2N} \\ \vdots \\ C_{i}A_{i}^{\left(n-1\right)N}\end{pmatrix}
\end{equation}
should have rank $n$ for all $i$. This is just the Kalman observability criterion for reconstructing $x_{i}$ from $y_{i}$ with $x_{i}$ evolving under $x_{i}\left(t+1\right)=A^{N}x_{i}\left(t\right)$ and $y_{i}\left(t\right)=C_{i}x_{i}$. 

\subsection{Exosystems Evolving in Continuous Time}
In the previous subsection, we had our signals produced by exosystems, but we assumed that those exosystems, and hence also our signals, evolve in discrete time. In the present subsection, we, instead, let our exosystems and signals evolve in continuous time. For this purpose, let the signals $x$ and $c$ be produced by the differential equations
\begin{align}
\dot{x}\left(t\right)&=A\hspace{1pt}x\left(t\right),\;\;\;x\left(0\right)=x_{0}, \label{eq:ODE} \\
\hspace{.75pt}\dot{c}\hspace{.75pt}\left(t\right)&=\hspace{.25pt}S\hspace{1pt}\hspace{.75pt}c\hspace{.25pt}\hspace{.75pt}\left(t\right),\;\;\;\hspace{.7pt}c\hspace{.7pt}\left(0\right)=\hspace{.15pt}c\hspace{1.25pt}_{0}, \label{eq:contcompressor}
\end{align}
respectively, with both $A$ and $S$ being real $n\times n$ matrices. Since we are interested in compression of periodic signals, further let $A$ be skew-symmetric, i.e., a member of the Lie algebra $\mathfrak{so}\left(n\right)$. A lossless compressor \eqref{eq:contcompressor} will allow us to reconstruct $x$ from the compressed signal \eqref{eq:compression}. Schematically, this compression scheme is depicted in Fig.\ \ref{fig:conttime}.

\begin{figure}
\centering
\begin{tikzpicture}
 \draw (0,0) rectangle (2,1) node [pos=.5] {$\dot{x}=A\hspace{1pt}x$};
 \draw (0,-1.5) rectangle (2,-.5) node [pos=.5] {$\dot{c}\hspace{.75pt}=\hspace{.25pt}S\hspace{.75pt}c$};
 \draw (3,-1.5) rectangle (4,1) node [pos=.5] {$\langle\cdot,\cdot\rangle$};
 \draw[->] (2,.5) -- (3,.5);
 \draw[->] (2,-1) -- (3,-1);
 \draw[->] (4,-.25) -- (5,-.25) node [anchor=west] {$y$};
\end{tikzpicture}
 \caption{The signal $x$ produced by the differential equation $\dot{x}=Ax$ is compressed by the signal $c$ produced by the differential equation $\dot{c}=Sc$.}\label{fig:conttime}
\end{figure}
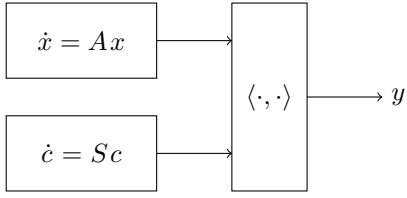

Our next result states that, for every skew-symmetric $A$, there exists a lossless compressor, uniquely defined by the matrix $S$ and the vector $c_{0}$. Further, that $S$ will be skew-symmetric as well. Thereby, we make use of the Cartan subalgebra, call it $\mathfrak{h}$, of $\mathfrak{so}\left(n\right)$, which is an abelian subspace of $\mathfrak{so}\left(n\right)$ consisting of block-diagonal matrices with $\mathfrak{so}\left(2\right)$ blocks on their diagonal and another $0$ attached in odd dimensions.
\begin{theorem}
Let $c$ and $x$, respectively, evolve under \eqref{eq:ODE} and \eqref{eq:contcompressor}. Then, for every $A\in\mathfrak{so}\left(n\right)$, there exist $S\in\mathfrak{so}\left(n\right)$ and $c_{0}\in\mathbb{R}^{n}$ such that $x$ can be uniquely reconstructed from $y$.
\end{theorem}
\begin{proof}
 For every skew-symmetric $A$, there is a conjugation $T\tilde{A}=\hspace{-2pt}AT$ bringing $A$ to the Cartan subalgebra $\mathfrak{h}$ of $\mathfrak{so}\left(n\right)$. Denote the $\mathfrak{so}\left(2\right)$ matrices on the diagonal of $\tilde{A}$ by $\tilde{A}_{i}$. Those matrices $\tilde{A}_{i}$ will be of the form  $\tilde{A}_{i}=\omega_{i}J$ with $J=\left(\begin{smallmatrix}\hphantom{-}0 & 1 \\ -1 & 0\end{smallmatrix}\right)$ denoting the generator of $\mathfrak{so}\left(2\right)$. Now choose real scalars $\theta_{i}$ such that all differences $\omega_{i}-\theta_{i}$ are pairwise distinct and nonzero and compose a matrix $\tilde{S}$ from $\mathfrak{h}$ by letting $\theta_{i}J$ being its $i$th diagonal block, with a zero attached in odd dimensions. Then choose $S$ via the conjugation $ST=T\tilde{S}$. Further, let 
\begin{equation}
 c_{0}=\begin{cases}T\left(e_{1}+e_{3}+\dots+e_{n-1}\right) & n\text{ even} \\
                    T\left(e_{1}+e_{3}+\dots+e_{n}\right) & n\text{ odd} \\
       \end{cases}
\end{equation}
such that the vectors $e^{\left(S-A\right)t}c_{0}$, $t\geq0$, span $\mathbb{R}^{n}$. Since $\tilde{A}$, $\tilde{S}$ commute and $A$, $S$ were brought to $\mathfrak{h}$ via the same conjugation, the matrices $A$ and $S$ commute as well. Thus, $y$ reads
 \begin{equation}
  y\left(t\right)=\langle e^{St}c_{0},e^{At}x_{0}\rangle
  =\langle e^{\left(S-A\right)t}c_{0},x_{0}\rangle
 \end{equation}
and we therefore conclude that $x_{0}$ can be uniquely reconstructed from $y$. All other $x\left(t\right)$ are reconstructed by left-multiplication with $e^{At}$.\end{proof}

The proof proceeded as follows: First, $S$ was chosen such that it commutes with $A$. Second, $c_{0}$ was chosen such that $e^{(S-A)t}c_{0}$, $t\geq0$, spans $\mathbb{R}^{n}$. Both constructions were realized through matrices, $\tilde{A}$ and $\tilde{S}$, from the Cartan subalgebra $\mathfrak{h}$. 

For instance, as a special case, $S$ could be the zero matrix when $x$ can be reconstructed from compression by a constant signal $c\left(t\right)=c_{0}$ (which is the case if all $\omega_{i}$ are pairwise distinct and nonzero). In situations where this is impossible, some nonzero and sufficiently rich $S$ will propel $c_{0}$ through $\mathbb{R}^{n}$, yielding a periodic compression signal $c$ that provides further directions of the vectors $e^{(S-A)t}c_{0}$, $t\geq 0$, which the signal $e^{-At}c_{0}$ alone did not attain.

As such, our result is constructive; given a skew-symmetric matrix $A$, our proof allows to design a periodic compressor \eqref{eq:contcompressor}. One first brings $A$ to $\tilde{A}\in\mathfrak{h}$ via conjugation with some matrix $T\in\operatorname{SO}\left(n\right)$ and then chooses real scalars $\theta_{i}$ such that all $\omega_{i}-\theta_{i}$ are pairwise distinct and nonzero, i.e.
\begin{equation}
 \omega_{i}-\theta_{i}\neq\omega_{j}-\theta_{j}\;\;\;\text{ and }\;\;\;\omega_{i}\neq\theta_{i},
\end{equation}
with $\omega_{i}J$ denoting the diagonal blocks of $\tilde{A}$. A periodic compressor is then constructed by arranging the blocks $\theta_{i}J$ along the diagonal of some matrix, $\tilde{S}$, lying in $\mathfrak{h}$, and then applying the inverse of the transformation which brought $A$ to $\mathfrak{h}$, yielding $S$. Similarly, $c_{0}$ is obtained by transforming the sum of all odd-indexed standard basis vectors via that transformation. 

\begin{example}
 Consider $n=7$ signals evolving under the differential equation \eqref{eq:ODE} where $A$ is taken to be the skew-symmetric matrix
\begin{equation}
A=\left(\begin{array}{rr|rr|rr|r}
& & 1 & 0 & & & \\ & &0 & 1  & & &\\ \hline -1 & 0 & &  & &\\ 0 & -1 & & & & \\ \hline & & & & 0 & 0\\& & & & 0 & 0\\ \hline & & & & & & 0
\end{array}\right)
\end{equation}
which is brought to $\mathfrak{h}$ via cyclic permutation of the row / column indices $2$ and $3$. The resulting matrix $\tilde{A}$ has the matrices $\omega_{1}J$, $\omega_{2}J$, $\omega_{3}J$ on its diagonal, with another $0$ appended since $n=7$ is odd. Here, $\omega_{1}=\omega_{2}=1$ and $\omega_{3}=0$. Thus, no constant signal $c$ would allow to reconstruct $x$ (either, two identical $\omega_{i}$ or one zero $\omega_{i}$ suffices for this), thus necessitating to choose $S$ nonzero and sufficiently rich. More specifically, we must choose real scalars $\theta_{i}$ rendering all $\omega_{i}-\theta_{i}$ pairwise distinct and nonzero, which is, e.g., the case for 
\begin{equation}
\theta_{1}=2,\;\;\;\;\;\theta_{2}=3,\;\;\;\;\;\theta_{3}=4. 
\end{equation}
Hence $\tilde{S}$ is constructed by arranging the matrices $2J$, $3J$, and $4J$ along its diagonal and appending an additional $0$. The matrix $S$ is then obtained by permuting the row / column indices $2$ and $3$. Similarly, $c_{0}$ should be chosen to be $e_{1}+e_{3}+e_{5}+e_{7}$, but with the order of the entries $2$ and $3$ exchanged, yielding 
\begin{equation}
 c_{0}=e_{1}+e_{2}+e_{5}+e_{7}.
\end{equation}
As a consequence, the vectors $e^{-At}e^{St}c_{0}=e^{(S-A)t}c_{0}$, $t\geq 0$, span $\mathbb{R}^{7}$, hence allowing to reconstruct $x$ from $y$.
\end{example}

 \begin{figure}[ht]
 \centering

 \begin{tikzpicture}[scale=.525]    
 \draw[white,line width=1.5] (4*3/5-.05,0*3/5-.05) rectangle (5*3/5+.05,1*3/5+.05); 
 \draw (1.5,3.5) -- (1.5,3.5) node {$t=0$};
          \draw[fill=gray!25!white] (0,0) rectangle (5*3/5,5*3/5);
        \draw[fill=black] (0*3/5,4*3/5) rectangle (1*3/5,5*3/5);
      \draw[fill=white] (1*3/5,4*3/5) rectangle (5*3/5,5*3/5);
 \draw (3/5,0) -- (3/5,3);
 \draw (2*3/5,0) -- (2*3/5,3);
 \draw (3*3/5,0) -- (3*3/5,3);
 \draw (4*3/5,0) -- (4*3/5,3);
 \draw (0,3/5) -- (3,3/5);
 \draw (0,2*3/5) -- (3,2*3/5);
 \draw (0,3*3/5) -- (3,3*3/5);
 \draw (0,4*3/5) -- (3,4*3/5);
 \draw[orange,line width=1.5] (0*3/5-.05,4*3/5-.05) rectangle (1*3/5+.05,5*3/5+.05); 
     \end{tikzpicture} 
  \begin{tikzpicture}[scale=.525]
     \draw (1.5,3.5) -- (1.5,3.5) node {$t=4$};
          \draw[fill=gray!25!white] (0,0) rectangle (5*3/5,5*3/5);
        \draw[fill=black] (0*3/5,4*3/5) rectangle (1*3/5,5*3/5);
      \draw[fill=white] (1*3/5,4*3/5) rectangle (5*3/5,5*3/5);
            \draw[fill=white] (0*3/5,3*3/5) rectangle (5*3/5,4*3/5);
                        \draw[fill=white] (0*3/5,2*3/5) rectangle (5*3/5,3*3/5);
        \draw[fill=black] (2*3/5,2*3/5) rectangle (3*3/5,3*3/5);
                                \draw[fill=white] (0*3/5,1*3/5) rectangle (5*3/5,2*3/5);
\draw[fill=black] (2*3/5,1*3/5) rectangle (3*3/5,2*3/5);
                                \draw[fill=white] (0*3/5,0*3/5) rectangle (5*3/5,1*3/5);
\draw[fill=black] (4*3/5,0*3/5) rectangle (5*3/5,1*3/5);
 \draw (3/5,0) -- (3/5,3);
 \draw (2*3/5,0) -- (2*3/5,3);
 \draw (3*3/5,0) -- (3*3/5,3);
 \draw (4*3/5,0) -- (4*3/5,3);
 \draw (0,3/5) -- (3,3/5);
 \draw (0,2*3/5) -- (3,2*3/5);
 \draw (0,3*3/5) -- (3,3*3/5);
 \draw (0,4*3/5) -- (3,4*3/5);
  \draw[orange,line width=1.5] (4*3/5-.05,0*3/5-.05) rectangle (5*3/5+.05,1*3/5+.05); 
     \end{tikzpicture}  \begin{tikzpicture}[scale=.525]  \draw[white,line width=1.5] (4*3/5-.05,0*3/5-.05) rectangle (5*3/5+.05,1*3/5+.05);    \draw (1.5,3.5) -- (1.5,3.5) node  {$t=8$};
          \draw[fill=gray!25!white] (0,0) rectangle (5*3/5,5*3/5);
      \draw[fill=white] (0*3/5,4*3/5) rectangle (5*3/5,5*3/5);
            \draw[fill=white] (0*3/5,3*3/5) rectangle (5*3/5,4*3/5);
              \draw[fill=black] (3*3/5,3*3/5) rectangle (4*3/5,4*3/5);
                        \draw[fill=white] (0*3/5,2*3/5) rectangle (5*3/5,3*3/5);
        \draw[fill=black] (2*3/5,2*3/5) rectangle (3*3/5,3*3/5);
                                \draw[fill=white] (0*3/5,1*3/5) rectangle (5*3/5,2*3/5);
\draw[fill=black] (3*3/5,1*3/5) rectangle (4*3/5,2*3/5);
 \draw (3/5,0) -- (3/5,3);
 \draw (2*3/5,0) -- (2*3/5,3);
 \draw (3*3/5,0) -- (3*3/5,3);
 \draw (4*3/5,0) -- (4*3/5,3);
 \draw (0,3/5) -- (3,3/5);
 \draw (0,2*3/5) -- (3,2*3/5);
 \draw (0,3*3/5) -- (3,3*3/5);
 \draw (0,4*3/5) -- (3,4*3/5);
  \draw[orange,line width=1.5] (3*3/5-.05,1*3/5-.05) rectangle (4*3/5+.05,2*3/5+.05);
     \end{tikzpicture}
     \begin{tikzpicture}[scale=.525]
\draw[white,line width=1.5] (4*3/5-.05,0*3/5-.05) rectangle (5*3/5+.05,1*3/5+.05);       \draw (1.5,3.5) -- (1.5,3.5) node  {$t=12$};
          \draw[fill=gray!25!white] (0,0) rectangle (5*3/5,5*3/5);
        \draw[fill=black] (4*3/5,4*3/5) rectangle (5*3/5,5*3/5);
      \draw[fill=white] (0*3/5,4*3/5) rectangle (5*3/5,5*3/5);
            \draw[fill=white] (0*3/5,3*3/5) rectangle (5*3/5,4*3/5);
              \draw[fill=black] (2*3/5,3*3/5) rectangle (3*3/5,4*3/5);
                        \draw[fill=white] (0*3/5,2*3/5) rectangle (5*3/5,3*3/5);
        \draw[fill=black] (2*3/5,2*3/5) rectangle (3*3/5,3*3/5);
 \draw (3/5,0) -- (3/5,3);
 \draw (2*3/5,0) -- (2*3/5,3);
 \draw (3*3/5,0) -- (3*3/5,3);
 \draw (4*3/5,0) -- (4*3/5,3);
 \draw (0,3/5) -- (3,3/5);
 \draw (0,2*3/5) -- (3,2*3/5);
 \draw (0,3*3/5) -- (3,3*3/5);
 \draw (0,4*3/5) -- (3,4*3/5);
   \draw[orange,line width=1.5] (2*3/5-.05,2*3/5-.05) rectangle (3*3/5+.05,3*3/5+.05);
     \end{tikzpicture}
     \begin{tikzpicture}[scale=.525]   \draw[white,line width=1.5] (4*3/5-.05,0*3/5-.05) rectangle (5*3/5+.05,1*3/5+.05);   \draw (1.5,3.5) -- (1.5,3.5) node {$t=16$};
          \draw[fill=gray!25!white] (0,0) rectangle (5*3/5,5*3/5);
      \draw[fill=white] (0*3/5,4*3/5) rectangle (5*3/5,5*3/5);
           \draw[fill=black] (2*3/5,4*3/5) rectangle (3*3/5,5*3/5);
            \draw[fill=white] (0*3/5,3*3/5) rectangle (5*3/5,4*3/5);
              \draw[fill=black] (1*3/5,3*3/5) rectangle (2*3/5,4*3/5);
 \draw (3/5,0) -- (3/5,3);
 \draw (2*3/5,0) -- (2*3/5,3);
 \draw (3*3/5,0) -- (3*3/5,3);
 \draw (4*3/5,0) -- (4*3/5,3);
 \draw (0,3/5) -- (3,3/5);
 \draw (0,2*3/5) -- (3,2*3/5);
 \draw (0,3*3/5) -- (3,3*3/5);
 \draw (0,4*3/5) -- (3,4*3/5);
    \draw[orange,line width=1.5] (1*3/5-.05,3*3/5-.05) rectangle (2*3/5+.05,4*3/5+.05);
     \end{tikzpicture}
    \caption{The original image at $t=0$ can be reconstructed from the measurements at $t=0,4,8,12,16$ due to coprimeness of $p=4$ and $n=5$.}\label{fig:smearrem}\vspace{-6pt}
 \end{figure}
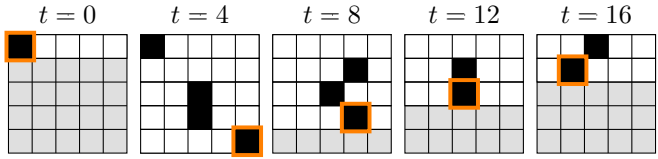

 \section{Removal of Smear}
 In this section, we return to our initial motivation of removing distortion effects caused by rolling shutter, i.e., we alleviate the effects of line-by-line readout of pixel data in digital cameras with active pixel sensors. A particular distortion effect that is regularly noticed due to rolling shutter has been coined \emph{smear} and occurs when rotating objects are photographed with a line-by-line pixel readout. We study the application of our techniques to the removal of smear in an extended example. To this end, consider the $5$-pixel rotor rotating in the $25$-pixel image depicted in the left column of Fig.\ \ref{fig:smear}. In the right column, we see a photography of the rotor subject to line-by-line readout of pixel data. When the entire image is read out at $t=4$, we find that the rotor is, indeed, distorted, and can hardly be recognized in the image. Yet, by virtue of our findings, since the evolution of pictures in the left column is $p=4$-periodic and our measurements in the right column switch through the lines of the picture $n=5$-periodically, it is possible to reconstruct the images on the left column from the images in the right column since the numbers $4$ and $5$ are coprime. We exemplify this fact by reconstructing the true picture at $t=0$. The first pixel of the rotor becomes visible at $t=0$ and at $t=4$, the fifth pixel is recorded. The remaining pixels $2$, $3$, and $4$ will become visible at the time instances $t=16$, $12$, and $8$, respectively, as it is illustrated in Fig.\ \ref{fig:smearrem}. The time instances $t=0,4,8,12,16$ which one requires for reconstructing our picture are precisely the solutions to the simultaneous congruences 
\begin{align}
 t&\equiv \hspace{.75pt}i\;\;\;\hspace{.75pt}\left(\operatorname{mod}5\right) \\ t&\equiv 0\;\;\;\left(\operatorname{mod}4\right)
\end{align} 
 with $i$ ranging from $0$ to $4$. In other words, pixel $\left(t\operatorname{mod}5\right)+1$ of the rotor becomes visible in the $t$th picture, where $t\operatorname{mod}4=0$.

\section{Discussion and Conclusion}
We proposed and studied a scheme for lossless and causal compression of real vector-valued periodic signals to real scalar-valued signals. In the simplest case, the compressed signal switches through the entries of the original signal in a periodic fashion. This compression was shown to be lossless whenever the periods of the periodic switch and the periodic signal are coprime, which, mathematically, amounts to a non-resonance condition on the discrete torus. This result was extended to compression with arbitrary periodic mixing signals. We presented a particular account on signals produced by exosystems, both in discrete and continuous time. Thereby, we, again, encountered non-resonance conditions for lossless compression.


 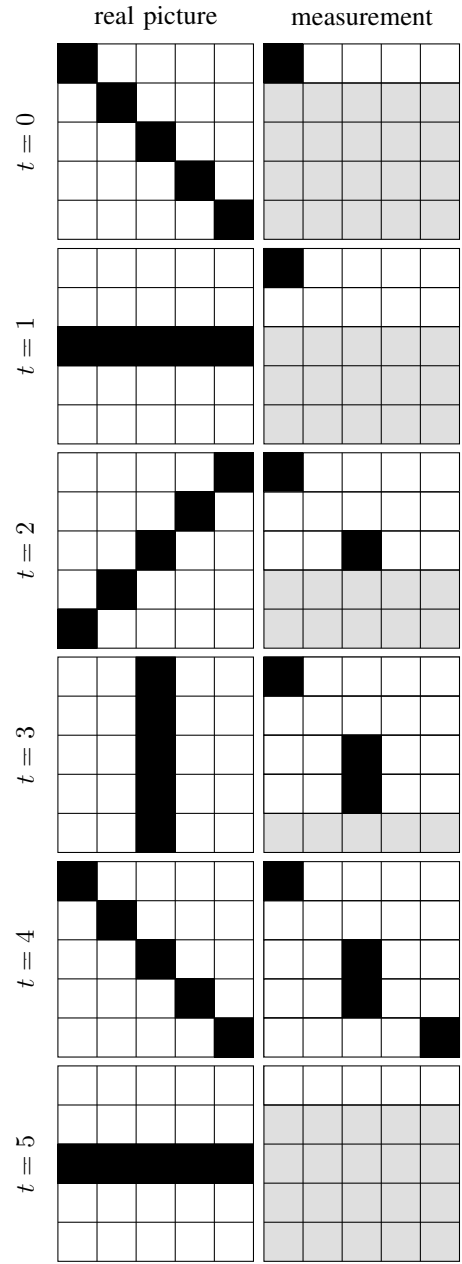
\begin{figure}[h!]
 \centering
   \begin{tikzpicture}[scale=.8666]
   \draw (1.5,3.4) -- (1.5,3.4) node {real picture};
   \draw (-.5,1.5) -- (-.5,1.5) node[rotate=90] {$t=0$};
 \draw (0,0) rectangle (3,3);
 \draw (3/5,0) -- (3/5,3);
 \draw (2*3/5,0) -- (2*3/5,3);
 \draw (3*3/5,0) -- (3*3/5,3);
 \draw (4*3/5,0) -- (4*3/5,3);
 \draw (0,3/5) -- (3,3/5);
 \draw (0,2*3/5) -- (3,2*3/5);	
 \draw (0,3*3/5) -- (3,3*3/5);
 \draw (0,4*3/5) -- (3,4*3/5);
 \filldraw (0*3/5,4*3/5) rectangle (1*3/5,5*3/5);
 \filldraw (1*3/5,3*3/5) rectangle (2*3/5,4*3/5);
 \filldraw (2*3/5,2*3/5) rectangle (3*3/5,3*3/5);
 \filldraw (3*3/5,1*3/5) rectangle (4*3/5,2*3/5);
 \filldraw (4*3/5,0*3/5) rectangle (5*3/5,1*3/5);
     \end{tikzpicture}   \begin{tikzpicture}[scale=.8666]
  \draw (1.5,3.4) -- (1.5,3.4) node {measurement};
          \draw[fill=gray!25!white] (0,0) rectangle (5*3/5,5*3/5);
        \draw[fill=black] (0*3/5,4*3/5) rectangle (1*3/5,5*3/5);
      \draw[fill=white] (1*3/5,4*3/5) rectangle (5*3/5,5*3/5);
 \draw (3/5,0) -- (3/5,3);
 \draw (2*3/5,0) -- (2*3/5,3);
 \draw (3*3/5,0) -- (3*3/5,3);
 \draw (4*3/5,0) -- (4*3/5,3);
 \draw (0,3/5) -- (3,3/5);
 \draw (0,2*3/5) -- (3,2*3/5);
 \draw (0,3*3/5) -- (3,3*3/5);
 \draw (0,4*3/5) -- (3,4*3/5);
     \end{tikzpicture} 
    
    \vspace{.3em}
    
     \begin{tikzpicture}[scale=.8666]
   \draw (-.5,1.5) -- (-.5,1.5) node[rotate=90] {$t=1$};
 \draw (0,0) rectangle (3,3);
 \draw (3/5,0) -- (3/5,3);
 \draw (2*3/5,0) -- (2*3/5,3);
 \draw (3*3/5,0) -- (3*3/5,3);
 \draw (4*3/5,0) -- (4*3/5,3);
 \draw (0,3/5) -- (3,3/5);
 \draw (0,2*3/5) -- (3,2*3/5);
 \draw (0,3*3/5) -- (3,3*3/5);
 \draw (0,4*3/5) -- (3,4*3/5);
 \filldraw (0*3/5,2*3/5) rectangle (1*3/5,3*3/5);
 \filldraw (1*3/5,2*3/5) rectangle (2*3/5,3*3/5);
 \filldraw (2*3/5,2*3/5) rectangle (3*3/5,3*3/5);
 \filldraw (3*3/5,2*3/5) rectangle (4*3/5,3*3/5);
 \filldraw (4*3/5,2*3/5) rectangle (5*3/5,3*3/5);
     \end{tikzpicture}  \begin{tikzpicture}[scale=.8666]
          \draw[fill=gray!25!white] (0,0) rectangle (5*3/5,5*3/5);
        \draw[fill=black] (0*3/5,4*3/5) rectangle (1*3/5,5*3/5);
      \draw[fill=white] (1*3/5,4*3/5) rectangle (5*3/5,5*3/5);
            \draw[fill=white] (0*3/5,3*3/5) rectangle (5*3/5,4*3/5);
 \draw (3/5,0) -- (3/5,3);
 \draw (2*3/5,0) -- (2*3/5,3);
 \draw (3*3/5,0) -- (3*3/5,3);
 \draw (4*3/5,0) -- (4*3/5,3);
 \draw (0,3/5) -- (3,3/5);
 \draw (0,2*3/5) -- (3,2*3/5);
 \draw (0,3*3/5) -- (3,3*3/5);
 \draw (0,4*3/5) -- (3,4*3/5);
     \end{tikzpicture}
 
        \vspace{.3em}

     \begin{tikzpicture}[scale=.8666]
 \draw (-.5,1.5) -- (-.5,1.5) node[rotate=90] {$t=2$};
\draw (0,0) rectangle (3,3);
\draw (3/5,0) -- (3/5,3);
\draw (2*3/5,0) -- (2*3/5,3);
\draw (3*3/5,0) -- (3*3/5,3);
\draw (4*3/5,0) -- (4*3/5,3);
\draw (0,3/5) -- (3,3/5);
\draw (0,2*3/5) -- (3,2*3/5);
\draw (0,3*3/5) -- (3,3*3/5);
\draw (0,4*3/5) -- (3,4*3/5);
\filldraw (0*3/5,0*3/5) rectangle (1*3/5,1*3/5);
\filldraw (1*3/5,1*3/5) rectangle (2*3/5,2*3/5);
\filldraw (2*3/5,2*3/5) rectangle (3*3/5,3*3/5);
\filldraw (3*3/5,3*3/5) rectangle (4*3/5,4*3/5);
\filldraw (4*3/5,4*3/5) rectangle (5*3/5,5*3/5);
    \end{tikzpicture}   \begin{tikzpicture}[scale=.8666]
         \draw[fill=gray!25!white] (0,0) rectangle (5*3/5,5*3/5);
       \draw[fill=black] (0*3/5,4*3/5) rectangle (1*3/5,5*3/5);
     \draw[fill=white] (1*3/5,4*3/5) rectangle (5*3/5,5*3/5);
           \draw[fill=white] (0*3/5,3*3/5) rectangle (5*3/5,4*3/5);
                       \draw[fill=white] (0*3/5,2*3/5) rectangle (5*3/5,3*3/5);
       \draw[fill=black] (2*3/5,2*3/5) rectangle (3*3/5,3*3/5);
\draw (3/5,0) -- (3/5,3);
\draw (2*3/5,0) -- (2*3/5,3);
\draw (3*3/5,0) -- (3*3/5,3);
\draw (4*3/5,0) -- (4*3/5,3);
\draw (0,3/5) -- (3,3/5);
\draw (0,2*3/5) -- (3,2*3/5);
\draw (0,3*3/5) -- (3,3*3/5);
\draw (0,4*3/5) -- (3,4*3/5);
    \end{tikzpicture}
   
    \vspace{.3em}

     \begin{tikzpicture}[scale=.8666]
  \draw (-.5,1.5) -- (-.5,1.5) node[rotate=90] {$t=3$};
\draw (0,0) rectangle (3,3);
\draw (3/5,0) -- (3/5,3);
\draw (2*3/5,0) -- (2*3/5,3);
\draw (3*3/5,0) -- (3*3/5,3);
\draw (4*3/5,0) -- (4*3/5,3);
\draw (0,3/5) -- (3,3/5);
\draw (0,2*3/5) -- (3,2*3/5);
\draw (0,3*3/5) -- (3,3*3/5);
\draw (0,4*3/5) -- (3,4*3/5);
\filldraw (2*3/5,0*3/5) rectangle (3*3/5,1*3/5);
\filldraw (2*3/5,1*3/5) rectangle (3*3/5,2*3/5);
\filldraw (2*3/5,2*3/5) rectangle (3*3/5,3*3/5);
\filldraw (2*3/5,3*3/5) rectangle (3*3/5,4*3/5);
\filldraw (2*3/5,4*3/5) rectangle (3*3/5,5*3/5);
    \end{tikzpicture}   \begin{tikzpicture}[scale=.8666]
         \draw[fill=gray!25!white] (0,0) rectangle (5*3/5,5*3/5);
       \draw[fill=black] (0*3/5,4*3/5) rectangle (1*3/5,5*3/5);
     \draw[fill=white] (1*3/5,4*3/5) rectangle (5*3/5,5*3/5);
           \draw[fill=white] (0*3/5,3*3/5) rectangle (5*3/5,4*3/5);
                       \draw[fill=white] (0*3/5,2*3/5) rectangle (5*3/5,3*3/5);
       \draw[fill=black] (2*3/5,2*3/5) rectangle (3*3/5,3*3/5);
                               \draw[fill=white] (0*3/5,1*3/5) rectangle (5*3/5,2*3/5);
\draw[fill=black] (2*3/5,1*3/5) rectangle (3*3/5,2*3/5);
\draw (3/5,0) -- (3/5,3);
\draw (2*3/5,0) -- (2*3/5,3);
\draw (3*3/5,0) -- (3*3/5,3);
\draw (4*3/5,0) -- (4*3/5,3);
\draw (0,3/5) -- (3,3/5);
\draw (0,2*3/5) -- (3,2*3/5);
\draw (0,3*3/5) -- (3,3*3/5);
\draw (0,4*3/5) -- (3,4*3/5);
    \end{tikzpicture}

    \vspace{.3em}

      \begin{tikzpicture}[scale=.8666]
   \draw (-.5,1.5) -- (-.5,1.5) node[rotate=90] {$t=4$};
 \draw (0,0) rectangle (3,3);
 \draw (3/5,0) -- (3/5,3);
 \draw (2*3/5,0) -- (2*3/5,3);
 \draw (3*3/5,0) -- (3*3/5,3);
 \draw (4*3/5,0) -- (4*3/5,3);
 \draw (0,3/5) -- (3,3/5);
 \draw (0,2*3/5) -- (3,2*3/5);
 \draw (0,3*3/5) -- (3,3*3/5);
 \draw (0,4*3/5) -- (3,4*3/5);
 \filldraw (0*3/5,4*3/5) rectangle (1*3/5,5*3/5);
 \filldraw (1*3/5,3*3/5) rectangle (2*3/5,4*3/5);
 \filldraw (2*3/5,2*3/5) rectangle (3*3/5,3*3/5);
 \filldraw (3*3/5,1*3/5) rectangle (4*3/5,2*3/5);
 \filldraw (4*3/5,0*3/5) rectangle (5*3/5,1*3/5);    \end{tikzpicture}   \begin{tikzpicture}[scale=.8666]
          \draw[fill=gray!25!white] (0,0) rectangle (5*3/5,5*3/5);
        \draw[fill=black] (0*3/5,4*3/5) rectangle (1*3/5,5*3/5);
      \draw[fill=white] (1*3/5,4*3/5) rectangle (5*3/5,5*3/5);
            \draw[fill=white] (0*3/5,3*3/5) rectangle (5*3/5,4*3/5);
                        \draw[fill=white] (0*3/5,2*3/5) rectangle (5*3/5,3*3/5);
        \draw[fill=black] (2*3/5,2*3/5) rectangle (3*3/5,3*3/5);
                                \draw[fill=white] (0*3/5,1*3/5) rectangle (5*3/5,2*3/5);
\draw[fill=black] (2*3/5,1*3/5) rectangle (3*3/5,2*3/5);
                                \draw[fill=white] (0*3/5,0*3/5) rectangle (5*3/5,1*3/5);
\draw[fill=black] (4*3/5,0*3/5) rectangle (5*3/5,1*3/5);
 \draw (3/5,0) -- (3/5,3);
 \draw (2*3/5,0) -- (2*3/5,3);
 \draw (3*3/5,0) -- (3*3/5,3);
 \draw (4*3/5,0) -- (4*3/5,3);
 \draw (0,3/5) -- (3,3/5);
 \draw (0,2*3/5) -- (3,2*3/5);
 \draw (0,3*3/5) -- (3,3*3/5);
 \draw (0,4*3/5) -- (3,4*3/5);
     \end{tikzpicture}

     \vspace{.3em}

      \begin{tikzpicture}[scale=.8666]
   \draw (-.5,1.5) -- (-.5,1.5) node[rotate=90] {$t=5$};
 \draw (0,0) rectangle (3,3);
 \draw (3/5,0) -- (3/5,3);
 \draw (2*3/5,0) -- (2*3/5,3);
 \draw (3*3/5,0) -- (3*3/5,3);
 \draw (4*3/5,0) -- (4*3/5,3);
 \draw (0,3/5) -- (3,3/5);
 \draw (0,2*3/5) -- (3,2*3/5);
 \draw (0,3*3/5) -- (3,3*3/5);
 \draw (0,4*3/5) -- (3,4*3/5);
 \filldraw (0*3/5,2*3/5) rectangle (1*3/5,3*3/5);
 \filldraw (1*3/5,2*3/5) rectangle (2*3/5,3*3/5);
 \filldraw (2*3/5,2*3/5) rectangle (3*3/5,3*3/5);
 \filldraw (3*3/5,2*3/5) rectangle (4*3/5,3*3/5);
 \filldraw (4*3/5,2*3/5) rectangle (5*3/5,3*3/5);
    \end{tikzpicture}   \begin{tikzpicture}[scale=.8666]
          \draw[fill=gray!25!white] (0,0) rectangle (5*3/5,5*3/5);
      \draw[fill=white] (0*3/5,4*3/5) rectangle (5*3/5,5*3/5);
 \draw (3/5,0) -- (3/5,3);
 \draw (2*3/5,0) -- (2*3/5,3);
 \draw (3*3/5,0) -- (3*3/5,3);
 \draw (4*3/5,0) -- (4*3/5,3);
 \draw (0,3/5) -- (3,3/5);
 \draw (0,2*3/5) -- (3,2*3/5);
 \draw (0,3*3/5) -- (3,3*3/5);
 \draw (0,4*3/5) -- (3,4*3/5);
     \end{tikzpicture}

    \caption{A 5-pixel rotor rotates in a 25-pixel picture (left) and is thereby recorded subject to pixel-by-pixel readout (right), thus causing ``smear''.}\label{fig:smear}\vspace{-6pt}
 \end{figure}

\bibliographystyle{IEEEtran}

\end{document}